\documentclass[journal]{IEEEtran}

\usepackage{amsmath}
\usepackage{amsthm}
\usepackage{color}
\usepackage{graphicx}
\usepackage[usenames,dvipsnames]{pstricks}
\usepackage{epsfig}
\usepackage{pst-grad} 
\usepackage{pst-plot} 
\usepackage{upgreek} 
\usepackage{flushend}
\usepackage{mathrsfs}
\usepackage{amssymb}
\usepackage{psfrag}

\usepackage{calc}
\usepackage{graphicx}
        \graphicspath{{figs/}{matlab/}{../figs/}}  

\newtheorem{proposition}{Proposition}

\newtheorem{theorem}{Theorem}

\usepackage{xspace}
\usepackage{bbm}
\input{mathlig}

\newcommand{\muspace}{\mspace{1mu}}

\DeclareRobustCommand{\scond}{\mathchoice{\muspace\vert\muspace}{\vert}{\vert}{\vert}}
\mathlig{|}{\scond}
\newcommand{\cond}{\mathchoice{\,\vert\,}{\mspace{2mu}\vert\mspace{2mu}}{\vert}{\vert}}

\DeclareRobustCommand{\discint}{\mathchoice{\mspace{-1.5mu}:\mspace{-1.5mu}}{\mspace{-1.5mu}:\mspace{-1.5mu}}{:}{:}}
\mathlig{::}{\discint}
\newcommand{\suchthat}{\mathchoice{\colon}{\colon}{:\mspace{1mu}}{:}}

\newcommand{\Ac}{\mathcal{A}}
\newcommand{\Bc}{\mathcal{B}}

\newcommand{\Dc}{\mathcal{D}}

\newcommand{\Gc}{\mathcal{G}}

\newcommand{\Jc}{\mathcal{J}}

\newcommand{\Sc}{\mathcal{S}}
\newcommand{\Tc}{\mathcal{T}}

\newcommand{\Wc}{\mathcal{W}}
\newcommand{\Xc}{\mathcal{X}}
\newcommand{\Yc}{\mathcal{Y}}
\newcommand{\Zc}{\mathcal{Z}}

\newcommand{\Rr}{\mathscr{R}}

\newcommand{\Lh}{{\hat{L}}}
\newcommand{\Mh}{{\hat{M}}}

\newcommand{\Ut}{{\tilde{U}}}

\def\b{\beta}

\def\e{\epsilon}
\def\eps{\epsilon}

\let\P\relax
\DeclareMathOperator\P{\textsf{P}}

\def\textiid{i.i.d.\@\xspace}
\newcommand\iid{\ifmmode\text{ i.i.d. } \else \textiid \fi}

\def\mathllap{\mathpalette\mathllapinternal}
\def\mathllapinternal#1#2{%
  \llap{$\mathsurround=0pt#1{#2}$}}

\def\clap#1{\hbox to 0pt{\hss#1\hss}}
\def\mathclap{\mathpalette\mathclapinternal}
\def\mathclapinternal#1#2{%
  \clap{$\mathsurround=0pt#1{#2}$}}

\let\oldstackrel\stackrel
\renewcommand{\stackrel}[2]{\oldstackrel{\mathclap{#1}}{#2}}

\renewcommand{\hbar}{h\mathllap{\overline{\vphantom{h}\hphantom{\rule{4.6pt}{0pt}}}\mspace{0.77mu}}}

\catcode`~=11 
\newcommand{\urltilde}{\kern -.06em\lower -.06em\hbox{~}\kern .02em}
\catcode`~=13 

\title{
Polar coding for interference networks}
\author{
\centerline{Lele Wang \quad and \quad Eren \c Sa\c so\u glu}
\thanks{L. Wang is with the Department of Electrical and Computer
Engineering, University of California, San Diego, La Jolla, CA 92093
USA (email: lew001@ucsd.edu).}
\thanks{E. \c Sa\c so\u glu is with the Department of Electrical
Engineering and Computer Sciences, University of California,
Berkeley, CA 94720 USA (email: eren@eecs.berkeley.edu).}
}

\begin{document}
\maketitle

\begin{abstract}
A polar coding scheme for interference networks is introduced.  The
scheme combines Ar\i kan's monotone chain rules for multiple-access
channels and a method by Hassani and Urbanke to `align' two
incompatible polarization processes.  It achieves the Han--Kobayashi
inner bound for two-user interference channels and generalizes to
interference networks. 
\end{abstract}

\section{Introduction}

Interference is one of the fundamental challenges in wireless
communication.  When multiple sender--receiver pairs communicate
simultaneously over a shared medium, the signal arrived at each receiver is a
mixture of its intended signal and undesired signals from all other
transmitters.
Therefore, even in the absence of noise, transmission from a sender to
its receiver is limited by the presence of transmission from other
parties.

The two-user memoryless interference channel models the simplest
such communication setting.  It is described by channel input
alphabets $\Xc$,~$\Wc$, output alphabets $\Yc$,~$\Zc$, and for all
$(x,w,y,z)\in\Xc\times\Wc\times\Yc\times\Zc$, the probability
$P(y,z|x,w)$ of receiving~$(y,z)$ when~$(x,w)$ are input to the
channel (Figure~\ref{fig:dmic}).

\begin{figure}[hbtp]
\footnotesize
\def\svgscale{1}
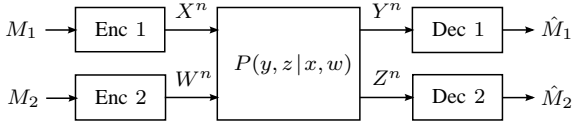
\caption{Two-user interference channels.}
 \label{fig:dmic}
\end{figure}

A $(2^{NR_1}, 2^{NR_2}, N)$ code for the two-user interference channel consists
of
\begin{itemize}
 \item two encoding functions $x^N(M_1)$ and
 $w^N(M_2)$, defined for messages $M_1\in[1::2^{NR_1}]$ and
 $M_2\in[1::2^{NR_2}]$. 
 \item two decoding functions $\hat{m}_1(y^N)$ and $\hat{m}_2(z^N)$, 
 defined for each received sequence $y^N\in\Yc^N$~and~$z^N\in\Zc^N$. 
\end{itemize}
Messages $M_1$ and $M_2$ are assumed to be uniformly distributed. 
The
average probability of error is defined as $P_e^{(N)} = \P\{(\Mh_1,
\Mh_2) \neq (M_1, M_2)\}$.
A rate pair $(R_1, R_2)$ is achievable if
there exists a sequence of $(2^{NR_1}, 2^{NR_2}, N)$ codes with 
$\lim_{N \to \infty} P_e^{(N)} = 0$. 
The capacity region is the
closure of the set of achievable rate pairs.

The capacity region of the two-user interference channel is not known
in general.  The best known inner bound to the capacity region was
given by Han and Kobayashi in~\cite{HK1981}.  Our aim here is to show
the achievability of this inner bound by polar coding techniques.

The Han--Kobayashi scheme consists in splitting each sender's message
into two parts, and letting each receiver decode one part of the
interfering sender's message in addition to both parts of its own
sender's message.  This creates a three-sender multiple-access channel
(MAC) for each receiver, and the code rates are required to satisfy
both MACs simultaneously.  Since these MACs share two of their three
senders, the situation is similar to a \emph{compound} setting, in
which codes must be designed to perform well simultaneously over
several (in this case, two) MACs.

Given these observations, one may hope to apply the standard results
on MACs to design polar codes for the interference channel.  In
particular, the corner points of a MAC's capacity region are known to
be achievable by standard polar coding techniques~\cite{STY2013}.
This readily implies the achievability of the entire MAC capacity
region by polar coding, since all achievable points can be turned into
corner points by the rate-splitting techniques of~\cite{GRUW2001}.
Unfortunately, rate-splitting techniques do not generalize in a
straightforward manner to the compound setting.  In particular, it is
shown in a parallel study~\cite{WSK2014} that standard applications of
rate-splitting techniques fail to achieve optimal compound rates in
general.  This makes it unclear whether polar coding techniques can be
combined with rate-splitting ideas to achieve the Han--Kobayashi inner
bound.  

Here, we show an alternative polar coding method that achieves the
capacity region of compound MACs and by extension the Han--Kobayashi
inner bound.  The method is based on appropriately combining two
techniques developed recently by Ar{\i}kan~\cite{Arikan2012}, and
Hassani and Urbanke~\cite{HU2013}.  We briefly review these techniques
first. 

\section{Preliminaries}

\subsection{Aligning polarized indices (\cite{HU2013})}
\label{sec:alignment}

Consider two binary-input memoryless channels  $P : X \to Y$ and $Q: X
\to Z$ with equal symmetric capacities $I(P)=I(Q)$.  Suppose we wish
to design a polar code that performs well over both of these channels.
For $N = 2^n$, define $U^N = X^N G_N$, where $G_N =
\bigl[\begin{smallmatrix} 1,0\\ 1,1\end{smallmatrix}\bigr]^{\otimes n}
B_N$ is the standard polar transformation. Here, ${\otimes n}$ denotes
the $n$th Kronecker power and $B_N$ is the `bit-reversal' permutation.
Define the channels $P_i:U_i\to Y^NU^{i-1}$ and $Q_i:U_i\to
Z^NV^{i-1}$ and sets
{\allowdisplaybreaks
\begin{align}
\begin{split}
\label{eqn:four-sets}
 \Gc_Y &= \{i\in[1::N]\suchthat I(P_i) > 1-2^{-N^\b}\},\\
 \Gc_Z &= \{i\in[1::N]\suchthat I(Q_i) > 1-2^{-N^\b}\},\\
 \Bc_Y &= \{i\in[1::N]\suchthat I(P_i) < 2^{-N^\b}\},\\
 \Bc_Z &= \{i\in[1::N]\suchthat I(Q_i) < 2^{-N^\b}\}.
\end{split}
\end{align}}%
for some $\b < 1/2$.  Standard polarization results imply that
$|\Gc_Y|/N\approx I(P)=I(Q)\approx|\Gc_Z|/N$ for large~$N$, and thus
almost all bit indices belong to one of the following four sets:
\allowdisplaybreaks{
\begin{align*}
 \Ac_\text{I} &= \Gc_Y \cap \Gc_Z,\\
 \Ac_\text{II} &= \Gc_Y \cap \Bc_Z,\\
 \Ac_\text{III}&= \Bc_Y \cap \Gc_Z,\\
 \Ac_\text{IV}&=  \Bc_Y \cap \Bc_Z.
\end{align*}}%
It suffices to discuss only the bit indices of the above four types,
and assume that the remaining bit values are fixed and revealed to all
receivers.  Note that type-I indices see clean channels for both
$P$~and~$Q$ and thus can carry information.  Similarly, type-IV indices are bad
for both channels and can be fixed.  Type-II~and~III indices are
\emph{incompatible}, i.e., they are good for one channel and bad for the other.
Moreover, the fraction $(|\Ac_\text{II}|+|\Ac_\text{III}|)/N$ of incompatible
indices is non-negligible in general~\cite{HKU2009}, and therefore standard
polar coding does not achieve the compound capacity of arbitrary channels
$P$~and~$Q$.  

Hassani and Urbanke propose a simple solution to this problem, which
\emph{aligns} the good indices of the two channels.  Given two independent
binary-input memoryless channels $P: X_1\to Y_1$ and
$Q:X_2\to Y_2$, define the binary-input channels
$$
(P,Q)^-(y_1,y_2| u_1)
	=\sum_{u_2} 
	\tfrac12 P(y_1| u_1\oplus u_2)Q(y_2| u_2),
$$
$$
(P,Q)^+(y_1,y_2,u_1|u_2)
	=\tfrac12 P(y_1| u_1\oplus u_2)Q(y_2| u_2),
$$
and note that
\begin{align}
\begin{split}
\label{eqn:minmax}
I((P,Q)^-)&\le \min\{I(P),\,I(Q)\}\\
I((P,Q)^+)&\ge \max\{I(P),\,I(Q)\}.
\end{split}
\end{align}
Now let $i$ and $j$ be a type-II and a type-III index, respectively.  That is,
\begin{align*}
I(P_i)&\approx1 \; \text{ and } \; I(P_j)\approx0,\\
I(Q_i)&\approx0 \; \text{ and } \; I(Q_j)\approx1.
\end{align*}
It then follows from~\eqref{eqn:minmax} that
\begin{align*}
 I((P_i,P_j)^-)&\approx 0 \;\; \text{ and } \;\; I((P_i,P_j)^+)
\approx 1,\\
 I((Q_i,Q_j)^-)&\approx 0 \;\; \text{ and } \;\; I((Q_i,Q_j)^+) \approx 1,
\end{align*}
%
In words, combining two incompatible indices results in an almost perfect `plus'
channel and almost useless `minus' channel, regardless of the underlying
channel.  This `aligns' the mutual informations for such indices.  Taking two
blocks of $U^N$, one can combine almost all type-II indices from one block with
type-III indices from the other block,
since~$|\Ac_\text{II}|/N\approx|\Ac_\text{III}|/N$.  More precisely, suppose
$\Ac_\text{II} = \{c_1,c_2,\ldots,c_m\}$ and $\Ac_\text{III}
=\{d_1,d_2,\ldots,d_n\}$, where the elements are written in increasing order.
Define $U^N=X^NG_N$ and $E^N=X_{N+1}^{2N}G_N$.  Then, combining $U_{c_j}$ with
$E_{d_j}$, $j=1,\dotsc,q=\min\{m,n\}$, and leaving the remaining symbols
uncombined yields the length-$2N$ sequence
\begin{equation*}
\begin{split}
\Ut^{2N}&=
\bigl(U^{c_1-1}, E^{d_1-1}, U_{c_1} \oplus
E_{d_1}, E_{d_1},\\
&\hspace{1.8em}\cdots\\
&\hspace{1.8em}U_{c_{q-1}+1}^{c_q-1}, E_{d_{q-1}+1}^{d_q-1}, 
U_{c_q}\oplus E_{d_q}, E_{d_q},\\
&\hspace{1.8em}U_{c_q+1}^N, E_{d_q+1}^N\bigr).
\end{split}
\end{equation*}
Then, the mutual informations of channels $\Ut_i\to Y^{2N}\Ut^{i-1}$ and
$\Ut_i\to Z^{2N}\Ut^{i-1}$ are aligned for the combined indices
$\Ut_i=U_{c_j}\oplus E_{d_j}$ and $\Ut_i=E_{d_j}$, and unchanged for the
remaining ones.  Note again that the indices in~$\Ac_\text{III}$ of the first
block and~$\Ac_\text{II}$ of the second block are not combined with each other
and remain incompatible.  This ensures that the combined indices are polarized
as desired. The fraction of incompatible indices is thus halved by this
alignment, to $(|\Ac_\text{II}|+|\Ac_\text{III}|)/2N$.  Recursively aligning the
indices~$k$ times in this fashion then reduces this fraction to
$(|\Ac_\text{II}|+|\Ac_\text{III}|)/2^kN$, and thus the rate $I(P)=I(Q)$ can be
achieved on both channels by picking a large $k$. 

\begin{figure}[hbtp]
\centering
\small
\def\svgscale{1.4}
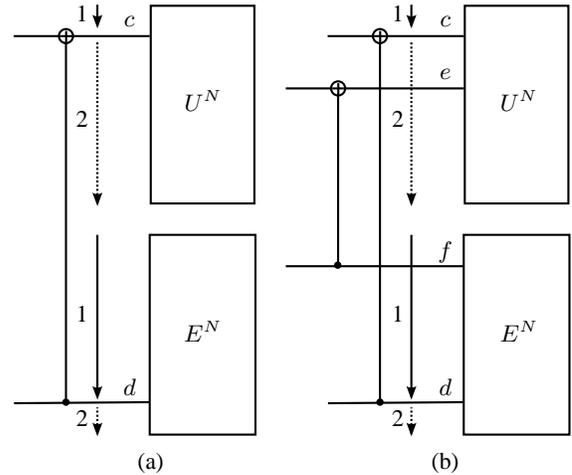
\caption{Alignment of the incompatible indices}
\label{fig:chainp2p}
\end{figure}
%

To show proper alignment of incompatible indices and its
corresponding decoding order, consider an example where $\Ac_\text{II} =
\{c\}$ and
$\Ac_\text{III} = \{d\}$. We
combine $U_c$ from block~1 with $E_d$ from block~2 as in
Figure~\ref{fig:chainp2p} (a). Decoding is thus done in the following
order. Variables along the solid line arrows should be decoded before the
variables along the dash line arrows. Variables along arrows with the same
number can be decoded parallelly. Figure~\ref{fig:chainp2p} (b) shows an example
where improper combining violates the successive decodability. Here in order to
decode $U_c \oplus E_d$, one needs to know $E^{d-1}$, and in particular $E_f$.
However, the decoding of $E_f$ involves $U_e$, which won't be available before
knowing $U_c \oplus E_d$. Therefore, it is crucial to order  type
II and type III indices in increasing order and combine the $j$-th type II
index from one block with the $j$-th type III index from another independent
block.


\subsection{`Polar Splitting' for MAC (\cite{Arikan2012})}
Consider a two-user MAC~$(\Xc \times \Wc, P(y|x,w), \Yc)$, where sender~1 and
sender~2 wish to communicate two messages $M_1$ and $M_2$ to the receiver by
respectively sending codewords $X^N(M_1)$ and $W^N(M_2)$ over $N$ uses
of the channel. The capacity region of this channel is given by 
\begin{align}
\label{eqn:mac-union}
\bigcup_p \Rr (p),
\end{align}
where the union is over all distributions of the form
$p=p(q)p(x|q)p(w|q)P(y|x,w)$, and $\Rr(p)$ is the set of non-negative rate pairs
$(R_1, R_2)$ satisfying
\begin{align}
\begin{split}
\label{eqn:pentagon}
R_1 &\le I(X; Y,W|Q),\\
R_2 &\le I(W; Y,X|Q),\\
R_1 + R_2 &\le I(X,W; Y|Q).
\end{split}
\end{align}
The subset of $\Rr(p)$ satisfying $R_1+R_2=I(X,W;Y|Q)$ is called its
\emph{dominant face}, and the two points $(I(X;Y|Q), I(W;Y,X|Q))$ and
$(I(X;Y,W|Q), I(W;Y|Q))$ are called its \emph{corner points}.  We will
first consider uniform $X$~and~$W$ and constant~$Q$; generalizations
to arbitrary distributions are discussed in Section~\ref{sec:coded-ts}.

In~\cite{Arikan2012}, Ar\i kan develops a polar coding method that achieves the
entire dominant face based on the following observations:
Let $U^N = X^N G_N$ and $V^N = W^N G_N$.
Consider the chain rules of the form
\[
\sum_{i=1}^{2N} I(S_i;Y^N|S^{i-1}),
\]
where $S^{2N} = (S_1,\ldots,S_{2N})$ is a \emph{monotone} permutation of
$U^NV^N$, i.e.,  elements of both $U^N$~and~$V^N$ appear in increasing order
in~$S^{2N}$.  Let $\Sc_U$~and~$\Sc_V$ respectively denote the set of indices of
$S^{2N}$ 
with $S_i=U_k$~and~$S_i=V_k$, and define the rates
\begin{align}
\begin{split}
\label{eqn:rates}
R_1 &= \frac{1}{N} \sum_{\substack{i\in \Sc_U}} I(S_i;Y^N|S^{i-1}),\\
R_2 &= \frac{1}{N} \sum_{\substack{i\in \Sc_V}} I(S_i;Y^N|S^{i-1}).
\end{split}
\end{align}
The entire region $\Rr(p)$ can be achieved by polar coding if $(R_1,R_2)$ can be
set to arbitrary values on the dominant face and if the mutual informations
$I(S_i;Y^N|S^{i-1})$ are polarized.  It turns out that these requirements are
satisfied by permutations of the form~$S^{2N}=(U^i,V^N,U_{i+1}^N)$.

\begin{proposition}[\cite{Arikan2012}]
\label{prop:2-split}
For every $\epsilon>0$, $\beta<1/2$, and rate pair $(I_1,I_2)$ on the dominant
face of $\Rr(p)$, there exist an $N$ and a permutation
$S^{2N}=(U^i,V^N,U_{i+1}^N)$
such that
\begin{itemize}
\item[(i)]
$|R_1 - I_1| < \eps$ and $|R_2 - I_2| < \eps$,
\item[(ii)]
$$
\frac{|\Gc^{(1)}|}{N} >R_1-\epsilon
\quad\text{ and }\quad 
\frac{|\Gc^{(2)}|}{N}>R_2-\epsilon,
$$
\end{itemize}
where 
\begin{align*}
\Gc^{(1)} = \{i\in \Sc_U 
	\suchthat I(S_i;Y^N|S^{i-1}) > 1-2^{-N^\beta}\},\\
\Gc^{(2)} = \{i\in \Sc_V 
	\suchthat I(S_i;Y^N|S^{i-1}) > 1-2^{-N^\beta}\}.
\end{align*}
\end{proposition}

\section{Two-user Compound MAC}

We are now ready to described a polar coding scheme for the two-user compound
MAC consisting of two channels $P_Y(y|x,w)$ and $P_Z(z|x,w)$.  The channel is
assumed to be known at the receiver but not at the transmitter.  A rate pair
$(R_1,R_2)$ is achievable if there exists a sequence of codes with rates
approaching $(R_1,R_2)$ and vanishing error probability on both MACs.  The
capacity region is described by
\begin{equation}
\label{eq:comp-mac}
\bigcup_p \bigl(\Rr_Y(p) \cap \Rr_Z(p)\bigr),
\end{equation}
where $\Rr_Y(p)$~and~$\Rr_Z(p)$ are as in~\eqref{eqn:pentagon}. 
Recall that for the simple MAC, the time-sharing random variable~$Q$
in~\eqref{eqn:pentagon} can be replaced by a convex hull operation on the union
in~\eqref{eqn:mac-union}.  
However, in the compound case,
this substitution leads in general to a strictly smaller rate region.

\subsection{Uniform Independent Inputs}
\label{sec:2-alignment}

\begin{figure}[hbtp]
\centering
\psfrag{r1}[lb]{$R_1$}
\psfrag{r2}[lb]{$R_2$}
\psfrag{corner}[lc]{target point $(I_1,I_2)$}
\includegraphics[width= .4\linewidth]{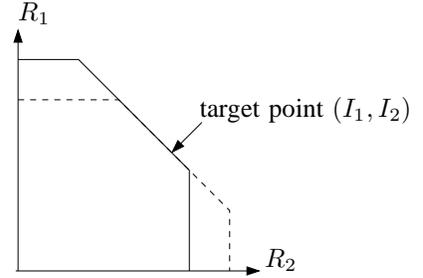}
\caption{Two MAC regions with equal sum-rate.}
\label{fig:equalsum}
\end{figure}

Assume $X$ and $W$ are uniform and independent, $Q = \emptyset$. The simplest
nontrivial case is when the two pentagons in~\eqref{eqn:pentagon}
intersect as in Figure~\ref{fig:equalsum}, with equal sum-rate $I(X,W;Y) =
I(X,W;Z)$. Let
$(I_1,I_2)$ be a rate point on the dominant face of this intersection. Let $U^N
= X^N G_N$ and $V^N = W^N G_N$.  By Proposition~\ref{prop:2-split}, there exists
an~$N$ and two monotone permutations $S^{2N}$ and
$T^{2N}$ for which the mutual informations
$I(S_i;Y^N|S^{i-1})$ and $I(T_i;Z^N|T^{i-1})$ are polarized, and the
corresponding rate pairs in~\eqref{eqn:rates} are close to $(I_1,I_2)$.
However, as in the point-to-point case, the two sets of mutual informations
$\{I(S_i;Y^N|S^{i-1})\suchthat i\in \Sc_U\}$ and
$\{I(T_i;Z^N|T^{i-1})\suchthat i\in \Tc_U\}$ may be incompatible. One can
similarly identify the type of index $i$ by comparing the mutual informations
of the bit-channels $S_j \to Y^NS^{j-1}$ and $T_k \to Z^NT^{k-1}$, where $S_j =
T_k = U_i$, and find the type II and
type III incompatible index sets $\Ac_\text{II}^{(1)}$ and
$\Ac_\text{III}^{(1)}$ for $U$, and $\Ac_\text{II}^{(2)}$ and
$\Ac_\text{III}^{(2)}$ for $V$.

\begin{figure}[hbtp]
\centering
\small
\def\svgscale{1.4}
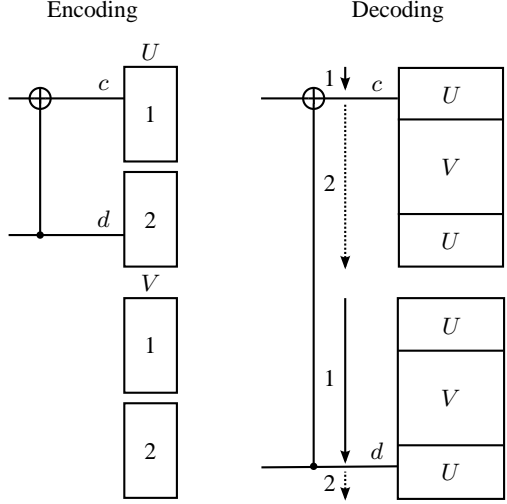
\caption{First recursion.}
\label{fig:chain}
\end{figure}

\begin{figure}[hbtp]
\centering
\small
\def\svgscale{1.3}
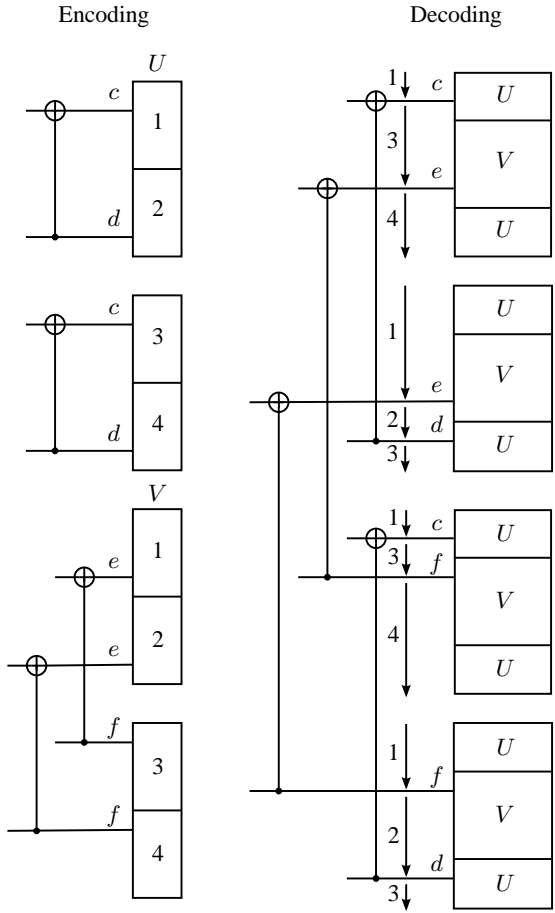

\caption{Second recursion.}
\label{fig:chain1x4}
\end{figure}

Now we apply the technique in Section~\ref{sec:alignment} to align
incompatible
indices of both $U$'s and $V$'s.  Here, as in the point-to-point case, care must
be taken to combine the random variables in a way that guarantees successive
decodability.  This can be done by aligning only the $U$'s or $V$'s in any given
recursion.  Also, as before, only half of the
incompatible indices of $U$'s (or $V$'s) are aligned in a single recursion.
Aligning the two index sets alternately over~$2k$ recursions, both fractions of
incompatible indices can be reduced to~$1/2^k$ times their original values. 
%

As an example, we show two recursions of alignment, where incompatible
$U$'s are aligned in the first recursion and incompatible $V$'s are aligned in
the second. Suppose $\Ac_\text{II}^{(1)} = \{c\}, \Ac_\text{III}^{(1)} =
\{d\},\Ac_\text{II}^{(2)} = \{e\},$ $\Ac_\text{III}^{(2)} = \{f\}$.
In the first recursion (blocks 1~and~2), $U$'s
are aligned while $V$'s are left uncombined (Figure~\ref{fig:chain}). The
decoding order for receiver~1 is shown on the right. After
stacking the $U$'s and $V$'s according to the monotone
permutation $S^{2N} = (U^i, V^N, U_{i+1}^N)$, decoding can be proceeded in a
similar fashion as in the
point-to-point alignment (recall Figure~\ref{fig:chainp2p}). In the second
recursion (Figure~\ref{fig:chain1x4}), a copy of
the length-$2N$ superblock is made  (blocks 3~and~4) for both $U$~and~$V$. The
two superblocks of
$V$'s are aligned while the two superblocks of $U$'s are left uncombined (as
shown on the left). At
the decoder~1, $U$'s and $V$'s from the same block are stacked according to the
monotone permutation $S^{2N}$. The uncombined indices in each block are decoded
until reaching a combined index. Then the two combined indices are decoded.
Since in each recursion, only incompatible indices for $U$ (or $V$) are combined
in the right order, successive decodability is guaranteed as in the
point-to-point case. More specifically for the running example, variables along
an arrow with smaller number should be decoded before those with a bigger
number, and variables along arrows with the same number can be decoded
parallelly. The monotone permutation $S^{8N}$ is defined by variables listed
according to such a decoding order. The corresponding rate pair $(R_1^s, R_2^s)$
are defined as before
\begin{align*}
 R_1^s &= \frac1N \sum_{i\in \Sc_U} I(S_i;Y^{4N}|S^{i-1}),\\
 R_2^s &= \frac1N \sum_{i\in \Sc_V} I(S_i;Y^{4N}|S^{i-1}).
\end{align*}
The decoding at the receiver~2 is performed 
according to the monotone permutation $T^{2N}$ in the similar fashion. The
resulting permutation $T^{8N}$ and its rate pair $(R_1^t, R_2^t)$ can be defined
similarly. Clearly, the  fraction of incompatible indices for $U$ (and $V$)
is halved in the first (second) recursion.


\begin{figure}[hbtp]
 \centering
 \psfrag{r1}[lb]{$R_1$}
 \psfrag{r2}[lb]{$R_2$}
 \psfrag{rate1}[rc]{$(I_1',I_2')$}
 \psfrag{rate2}[lc]{$(I_1'',I_2'')$}
 \psfrag{target}[lc]{$(I_1,I_2)$}
 \includegraphics[width= .4\linewidth]{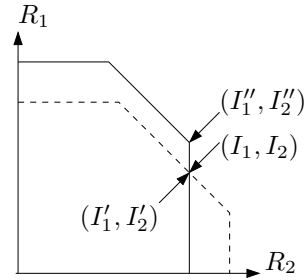}
 \caption{Two MAC regions with unequal sum-rates.}
 \label{fig:arbi-sum}
 \vspace{-1em}
\end{figure}

To achieve a rate point $(I_1,I_2)$ in the general case where $I(X,W;Y) \neq
I(X,W;Z)$ as in Figure~\ref{fig:arbi-sum}, one can find two monotone
permutations, which
respectively approximate rate pairs $(I_1',I_2')$ on the dominant face of
$P_Y(y|x,w)$ and $(I_1'',I_2'')$ on the dominant face of $P_Z(z|x,w)$, such that
\begin{align*}
 I_1 &\le \min\{I_1',I_1''\},\\
 I_2 &\le \min\{I_2',I_2''\}.
\end{align*}
Then, applying the approach above achieves the target rate point.

\subsection{Arbitrary Inputs}
\label{sec:coded-ts}
Based on the polar coding scheme developed for uniform and independent
$X$~and~$W$, one can adapt the method in~\cite[Section III-D]{STA2009} to design
a polar coding scheme for independent \emph{nonuniform} $X$ and $W$.
For correlated input distribution $\tilde{p} = p(q)p(x|q)p(w|q)$, there
exist $(X',W',Q)$
mutually independent and functions $x(x',q)$ and $w(w',q)$ that induce
the same
distribution as $\tilde{p}$. Now consider a
new MAC with inputs $X'$ and $W'$, vector output $(Y,Q)$, and
conditional distribution $P'(y,q|x',w') = p(q)P(y|x(x',q),w(w',q))$, where $Q$
is the
common randomness shared at the senders and the receiver. Then the achievable
rate region for the new MAC is
the set of rate pairs $(R_1,R_2)$ such that
\begin{align*}
 R_1 &\le I(X';Y,Q,W') = I(X;Y,W|Q),\\
 R_2 &\le I(W';Y,Q,X') = I(W;Y,X|Q),\\
 R_1 + R_2 &\le I(X',W';Y,Q) = I(X,W;Y|Q)
\end{align*}
for distribution
$p' = p(q)p(w')p(x')p(x|x',q)p(w|w',q)$ $P(y|x(x',q),w(w',q))$, where
$p(x|x',q)$ and $p(w|w',q)$ are $\{0,1\}$-valued according to $x(x',q)$ and
$w(w',q)$. This rate region is identical to $\Rr_Y(p)$ as $p' \equiv
p$. Similarly the rate region $\Rr_Y(p) \cap \Rr_Z(p)$ can be
described by considering the compound MAC with inputs $X'$ and $W'$,
vector output $(Y,Z,Q)$, and conditional distribution $P'(y,z,q|x',w') =
p(q)P_Y(y|x(x',q),w(w',q))P_Z(z|x(x',q),w(w',q))$. One can apply the method
designed for independent input to achieve arbitrary point in the rate region of
the new compound MAC.
To complete the proof, one just need to show the
existence of a good common random sequence $q^n$, which is shared at
the senders and the receiver before the transmission.  This is
guaranteed since the average probability of error over all possible
choices of $q^n$ is small.

\subsection{Main Result}
\begin{theorem}
\label{thm:mac}
For every $\epsilon>0$, $\beta<1/2$, and rate pair $(I_1,I_2)$ in the rate
region $\Rr_Y(p)\cap \Rr_Z(p)$, there exist $N, M = 2^kN$, and two monotone
permutation 
$S^{2M}$ and $T^{2M}$ with associated  rate pairs $(R_1^{s},R_2^s)$ and
$(R_1^t, R_2^t)$  such that for $j = 1,2$,
\begin{itemize}
\item[(i)]
$$|\min\{R_j^s, R_j^t\} - I_j| < \eps,$$
\item[(ii)]
$$
\frac{|\Gc_Y^{(j)}\cap \Gc_Z^{(j)}|}{M} >\min\{R_j^s,
R_j^t\}-\epsilon,
$$
where 
\begin{align*}
\Gc_Y^{(1)} &= \{i\in \Sc_U \suchthat I(S_i;Y^{M}|S^{i-1}) >
1-2^{-N^\beta}\},\\
\Gc_Z^{(1)} &= \{i\in \Tc_U \suchthat  I(T_i;Z^{M}|T^{i-1}) >
1-2^{-N^\beta}\},
\end{align*}
$\Gc_Y^{(2)}, \Gc_Z^{(2)}$ are defined similarly by replacing $U$ by
$V$.
\end{itemize}
\end{theorem}

%

The above theorem implies that arbitrary point in the capacity
region of the two-user compound MAC is achievable with the proposed polar
coding scheme. 
In the two-user \emph{strong} interference channel, 
that is $I(X;Y,W) \le I(X;Z,W)$ and $I(W;Z,X) \le I(W;Y,X)$ for all
$p(x)p(w)$, decoding both messages at each receiver is optimal and the
two-user compound MAC region coincides with the capacity region of the
interference channel. Therefore, the same technique applies to the two-user
strong interference channels.

\section{Interference Networks}
\label{sec:networks}

Now we generalize the result to $K$-sender $L$-receiver interference
networks with input alphabets $\Xc_1, \ldots, \Xc_K$, and output alphabets
$\Yc_1, \ldots, \Yc_L$, and conditional distribution $P(y^L|x^K)$
as depicted in Figure~\ref{fig:network}.
\begin{figure}[htbp]
\small
\def\svgscale{1.1}
\hspace{-1em}
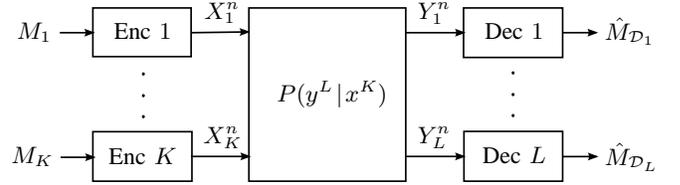
\caption{$K$-sender $L$-receiver interference networks.}
 \label{fig:network}
\end{figure}
Each sender~$j \in[1::K]$
communicates an independent message $M_j$ at rate $R_j$ and each
receiver~$l \in[1::L]$ wishes to recover a subset $\Dc_l \subseteq [1::K]$ of
the messages. The optimal rate region when the encoding is restricted to random
coding ensembles with superposition coding and time sharing~\cite{BEK2012} is
the
union over $\{(\Ac_1,\ldots,\Ac_L): \Ac_l \supseteq \Dc_l, l\in[1::L]\}$ of the
region
\begin{equation}
\label{eq:network}
 \bigcap_{l\in[1::L]}\Rr_{\Ac_l}(p),
\end{equation}
where the input distribution is of the form $p =
p(q)$ $\bigl(\prod_{j=1}^Kp(x_j|q)\bigr)P(y^L|x^K)$ and $\Rr_{\Ac_l}(p)$ is
the set of rate tuples $(R_j\suchthat j \in \Ac_l)$ such that
\begin{equation}
 \label{eq:macregion}
 R(\Jc) \le I(X_\Jc;Y_l,X_{\Ac_l\setminus \Jc}|Q) 
\end{equation}
for all $\Jc \subseteq \Ac_l$. Here we introduce notation 
$$R(\Jc) := \sum_{j\in \Jc} R_j$$
and 
$$X_\Jc := (X_j\suchthat j \in \Jc)$$
for an index set $\Jc$. It is clear from~\eqref{eq:network} that this rate
region is also a compound MAC region.

To apply the proposed polar coding scheme to the interference networks, one
needs to (i) generalize Ar\i kan's polar
splitting result to $K$-user MAC and (ii) align more than two incompatible
polarization processes, each of which involves codes from $K$ users. We prove
(i) in Section~\ref{sec:k-mac} and discuss (ii) in
Section~\ref{sec:k-alignment}. We show two important applications in
Sections~\ref{sec:hk} and~\ref{sec:bc}.

\subsection{`Polar Splitting' for $K$-user MAC}
\label{sec:k-mac}

Consider a $K$-user MAC, where transmitter~$j, j \in [1::K]$, wishes to
communicate a message $M_j$ reliably to the receiver by sending a codeword
$X_j^N(M_j) = (X_{j1}, X_{j2}, \ldots, X_{jN})$ over the memoryless channel
$P(y|x^K)$. The receiver wishes to recover all the messages $M_{[1::K]}$.
The capacity region of the $K$-user MAC is described by
\[
 \bigcup_p \Rr_{[1::K]}(p),
\]
where the union is over all distributions of the form
$p=p(q)\big(\prod_{i=1}^Kp(x_i|q)\big) P(y|x^K)$, and
$\Rr_{[1::K]}(p)$ is defined as in~\eqref{eq:macregion}.


Let $U_j^N = X_j^N G_N$ for $j \in [1::K]$. Similar to the two-user MAC case,
we have the chain rule of the form
\[
 \sum_{i = 1}^{KN} I(S_i;Y^N|S^{i-1}),
\]
where $S^{KN}$ is a \emph{monotone} permutation of $(U_1^N, \ldots, U_K^N)$,
i.e., elements of $U_j^N$ appear in increasing order in $S^{KN}$ for all 
$j \in [1::K]$. Let $\Sc_j$ denote the index set $\{i\suchthat S_i = U_{jk}
\text{ for some } k\}$. Define the associated rate tuple $(R_1, \ldots, R_K)$ of
the monotone permutation as
\[
 R_j = \frac{1}{N}\sum_{i\in \Sc_j} I(S_i;Y^N|S^{i-1})
\]
for $j \in [1::K]$.
We now generalize Ar\i kan's polar-splitting result to $K$ users. 

\begin{proposition}
\label{prop:k-split}
For every $\epsilon>0, \b < 1/2$, and rate tuple $(I_1, \ldots, I_K)$ on
the dominant
face of $\Rr_{[1::K]}(p)$, there exists an $N$ and a monotone permutation
$S^{KN}$
 such that for all $j \in [1::K]$,
\begin{itemize}
 \item[(i)] $$\lvert R_j-I_j\rvert\le\epsilon,$$
\item[(ii)] 
\[
 \frac{|\Gc^{(j)}|}{N} > R_j-\e,
\]
where 
$$\Gc^{(j)} = \{i \in \Sc_j \suchthat I(S_i;Y^N|S^{i-1}) >
1-2^{-N^\b}\}.$$
\end{itemize}
\end{proposition}

\begin{proof}
We prove statement (i) by induction. The case $K=2$ holds by
Proposition~\ref{prop:2-split}. Suppose the statement holds up to $K-1$. We
prove the statement for $K$.

Assume without loss of generality that we
start by decoding $U_1^{i_0}$ for some $i_0 \in [1::N]$. We specify $i_0$ by the
following procedure. Let $i$ increase from $0$ to $N$ and consider the
quantities
\begin{equation}
\label{eq:ui}
 \tfrac1N I(U^N_\Jc;Y^N,U_1^i)
\end{equation}
for each $\Jc \subseteq [2::N]$.
Some observations follow: 
\begin{enumerate}
 \item As $i$ increases, each mutual information term
increases by at most $1/N$ in each step, since the increment is
$I(U_{1i};U^N_\Jc|Y^N,U_1^{i-1})/N \le 1/N$.
 \item There exists an $i$ such that for at least one $\Jc \subseteq [2::K]$,
the following is violated
 \begin{align}
 \tfrac{1}{N}I(U^N_\Jc;Y^N,U_1^i) &< I(\Jc). \label{eq:rj}
\end{align}
\end{enumerate}
To see 2), set $U_1^i = \emptyset$ and $U_1^i = U_1^N$ respectively.
We have
\begin{align*}
 \tfrac1N I(U^N_\Jc;Y^N) &\le I(\Jc) \quad\text{for } \Jc \subset [2::K],\\
 \tfrac1N I(U^N_{[2::K]};Y^N) &\le I([2::K]) \le I(U^N_{[2::K]};Y^N,U_1^N).
\end{align*}
As $i$ increases, the mutual information terms in~\eqref{eq:ui} increase
steadily. Therefore,
there exists an $i$ such that~\eqref{eq:rj} is violated for some $\Jc \subseteq
[2::K]$.  Take the smallest such $i$ as $i_0$.

Suppose at $i = i_0$, the inequality in~\eqref{eq:rj} is violated at $\Jc_0$ for
the first time. As the increment on the left-hand-side of~\eqref{eq:rj} is
bounded by $1/N$, we roughly have
\begin{equation}
 \label{eq:j0}
 \tfrac1N I(U^N_{\Jc_0};Y^N,U_1^{i_0}) = I(\Jc_0).
\end{equation}
This divides the $K$-dimensional rate-approximation into two subproblems of
smaller dimensions. 

{\it Problem 1:} For $\Jc \subseteq \Jc_0$, we have
\[
 \tfrac1N I(U^N_\Jc;Y^N,U_1^{i_0}) < I(\Jc) \quad\text{for all } \Jc \subset
\Jc_0
\]
and
\[
 \tfrac1N I(U^N_{\Jc_0};Y^N,U_1^{i_0}) = I(\Jc_0).
\]
This is a rate-approximation problem for the rate tuple
$(I_j\suchthat j \in \Jc_0)$ on the dominant face of a $|\Jc_0|$-user MAC with
output
$(Y^N,U_1^{i_0})$. 

{\it Problem 2:} For all $\Jc
\supseteq \Jc_0$, we subtract~\eqref{eq:j0} from~\eqref{eq:rj}.
Let $\Tc = \Jc \setminus \Jc_0$,
$\Tc_0 = [2::K]\setminus \Jc_0$, and $I_1'= I_1 -\frac1N I(U_1^{i_0};Y^N)$.
This yields
\begin{align*}
 \tfrac1N I(U^N_\Tc;Y^N,U^N_{\Jc_0},U_1^{i_0}) &\le I(\Tc),\\
 \tfrac1N I(U^N_\Tc,U_{1,i_0+1}^N;Y^N,U^N_{\Jc_0},U_1^{i_0}) &\le
I(\Tc) +I_1',\\ 
\tfrac1N I(U^N_{\Tc_0},U_{1,i_0+1}^N;Y^N,U^N_{\Jc_0},U_1^{i_0}) &= I(\Tc_0) +
I_1'.
\end{align*}
This is a rate-approximation problem for the rate tuple
$(I_1',(I_j\suchthat j \in \Tc_0))$ on the dominant face of a $(K-|\Jc_0|)$-user
MAC with
output $(Y^N,U_1^{i_0},U^N_{\Jc_0})$. 

Note that $1\le|\Jc_0|\le K-1$. Thus
both problems are reduced to a smaller dimension. The final path is obtained by
cascading $1^{i_0}$, $b^{|\Jc_0|N}$ (the solution from problem~1), and
$b^{KN-|\Jc_0|N-i_0}$ (the solution from problem~2). 

The polarization result (ii) is obtained by standard path scaling as
in~\cite{Arikan2012}. This concludes the proof.
\end{proof}

\subsection{Aligning polarized indices for $K$ users}
\label{sec:k-alignment}
Suppose we have two monotone permutations for two $K$-user MACs. To align the
incompatible indices for all users, one can continue the method
in Section~\ref{sec:2-alignment} and sequentially align the incompatible
indices for each $U_j^N, j\in[1::K]$. After alternately aligning $K$ index sets
over $Km$ recursions, the fraction of the incompatible indices for each user is
reduced to $1/2^m$ times the original fraction.
The method for aligning $L$ monotone permutations can be done by recursively
aligning two permutations as in~\cite{HU2013}.

\subsection{Han--Kobayashi Inner bound}
\label{sec:hk}

As an important special case, we show how the scheme above can be used to
achieve the Han--Kobayashi inner bound, the best known inner bound for general
two-user interference channels $P(y_1,y_2|x_1,x_2)$.

\begin{figure}[hbtp]
\footnotesize
\def\svgscale{1.2}
\hspace{-1em}
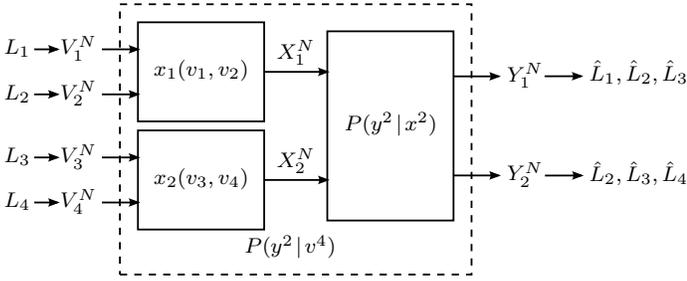
\caption{Han--Kobayashi coding scheme.}
 \label{fig:hk}
\end{figure}

The Han--Kobayashi coding scheme is illustrated in Figure~\ref{fig:hk}. Message
$M_1$ is split into two independent parts $(L_1,L_2)$ and message $M_2$ is
split into two independent parts $(L_3,L_4)$. Message $L_j, j\in[1::4]$, is
carried by codeword $V_j^N(L_j)$. Then the
channel inputs $X_1^N$ and $X_2^N$ are formed using two symbol-by-symbol
mappings
$x_1(v_1,v_2)$ and $x_2(v_3,v_4)$. Receiver~1 uniquely
decodes
$(\Lh_1, \Lh_2, \Lh_3)$ upon receiving $Y_1^N$, while receiver~2 uniquely
decodes
$(\Lh_2, \Lh_3, \Lh_4)$ upon receiving $Y_2^N$. The achievable rate region of
the
Han--Kobayashi coding scheme is given by
\begin{equation}
\label{eq:hk}
 \bigcup_p \text{Proj}_{4\to 2} \bigl(\Rr_1(p) \cap \Rr_2(p)\bigr).
\end{equation}
Here the input distribution is of the form $p =  p(q)$ $ \bigl(\prod_{j=1}^4
p(v_j|q)\bigr)p(x_1|v_1,v_2,q)p(x_2|v_3,v_4,q)P(y_1,y_2|x_1,x_2)$, where
$p(x_1|v_1,v_2,q)$
and $p(x_2|v_3,v_4,q)$ are $\{0,1\}$-valued according to functions
$x_1(v_1,v_2,q)$ and $x_2(v_3,v_4,q)$. The rate region $\Rr_1(p)$ is the set
of rate triples $(R_1',R_2',R_3')$ such that
\begin{align*}
 R_\Jc' \le I(V(\Jc);Y_1^N,V([1::3]\setminus \Jc)|Q)
\end{align*}
for all $\Jc \subseteq [1::3]$.
The rate region $\Rr_2(p)$ is the set of rate triples $(R_2', R_3',R_4')$
such that 
\begin{align*}
 R_\Jc' \le I(V(\Jc);Y_2^N,V([2::4]\setminus \Jc)|Q)
\end{align*}
for all $\Jc \subseteq [2::4]$.
The operator $\text{Proj}_{4\to 2}$ is to apply the Fourier--Motzkin elimination
from the 4-dimensional space $(R_1',R_2',R_3',R_4')$ to the
2-dimensional space $(R_1,R_2)$ by setting $R_1 = R_1' + R_2'$ and $R_2 = R_3' +
R_4'$. 

It is clear from the Han--Kobayashi coding scheme that for each pair of
functions $x_1(v_1,v_2)$ and $x_2(v_3,v_4)$, the message splitting
transforms the the original two-user interference channel into a four-sender
two-receiver interference networks
\[
 P(y^2|v^4) = P(y_1,y_2|x_1(v_1,v_2),x_2(v_3,v_4)),
\]
where sender~$j \in \{1,2,3,4\}$
communicates an independent message $L_j$ at rate $R_j'$, receiver~1 recovers
the subset $\Dc_1 = \{1,2,3\}$ of the four messages, and receiver~2 recovers
the sucset $\Dc_2 = \{2,3,4\}$ of the four messages.

Note from expression~\eqref{eq:hk} that the auxiliary rate region
$(R_1',R_2',R_3',R_4')$ is the intersection of two 3-dimensional MAC regions,
two dimensions of which are in common. Therefore, one just needs to find two
monotone permutations that achieves any target point in the two MACs
respectively and sequentially align the two codes shared in common  using
the method
in Section~\ref{sec:2-alignment}.


\subsection{Superposition Coding for Broadcast Channels}
\label{sec:bc}
The method for interference networks also implies the achievability of the
superposition coding inner bound for general broadcast channels. As the simplest
example, consider a two-user broadcast channel $P(y_1,y_2|x)$, where the sender
wishes to communicate message $M_1$ to receiver~1 and message $M_2$ to
receiver~2. 

\begin{figure}[hbtp]
\footnotesize
\def\svgscale{1.2}

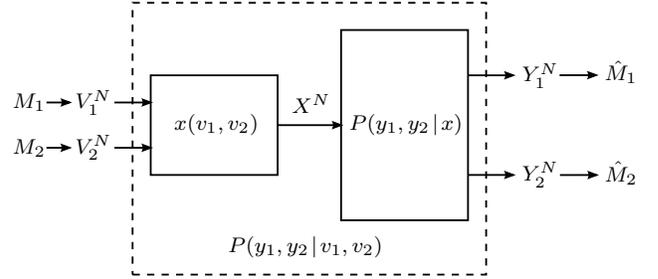
\caption{Cover's homogeneous superposition coding.}
 \label{fig:bc}
\end{figure}

Cover's homogeneous superposition coding~\cite{Cover1972} is illustrated in
Figure~\ref{fig:bc} . Two auxiliary sequences of codewords $V_1^N(M_1)$ and
$V_2^N(M_2)$ are generated according to independent distribution $p(v_1)p(v_2)$.
Then the channel input $X^N$ is formed through the symbol-by-symbol mapping $
x(v_1, v_2)$. This transforms the broadcast
channel into a two-sender two-receiver interference networks 
\[
 P(y_1,y_2|v_1,v_2) = P(y_1,y_2|x(v_1,v_2)),
\]
where sender~$j \in \{1,2\}$ communicates an independent message $M_j$ and
receiver~$j \in \{1,2\}$ recovers a subset $\Dc_j = \{j\}$ of the messages.
Clearly, this is another special case of the interference networks. 

It is worth mention that compared to Bergmans's heterogeneous superposition
coding~\cite{Bergmans1973}, where the codeword $X^N(M_1,M_2)$ is
generated conditioned on the codeword $V^N(M_1)$ according to distribution
$p(v)p(x|v)$, Cover's homogeneous superposition coding 
achieves in general a strictly larger rate region in the two-user
broadcast channels under optimal decoding~\cite{WSBK2013}. The rate region
achievable by Cover's superposition encoding  and optimal
decoding is~\cite{WSBK2013}
\[
 \bigcup_p \bigcup_{i=1}^4 \bigl(\Rr_{1i}(p) \cap \Rr_{2i}(p)\bigr),
\]
where the distribution is of the form $p =p(v_1)p(v_2)$ $p(x|v_1,v_2)
P(y_1,y_2|x)$ with $\{0,1\}$-valued $p(x|v_1,v_2)$ and 
$\Rr_{1i}(p)\cap \Rr_{2i}(p)$ corresponds to the rate region when the
decoders are required to uniquely recover the following message sets
\begin{align*}
i = 1 &\suchthat \Ac_1 = \{1\}, \Ac_2 = \{2\};\\
i = 2 &\suchthat \Ac_1 = \{1,2\}, \Ac_2 = \{2\};\\
i = 3 &\suchthat \Ac_1 = \{1\}, \Ac_2 = \{1,2\};\\
i = 4 &\suchthat \Ac_1 = \{1,2\}, \Ac_2 = \{1,2\}.
\end{align*}
For example, $\Rr_{13}(p) \cap \Rr_{23}(p)$ is the set of $(R_1,R_2)$ such
that
\begin{align*}
 R_1 &< I(V_1;Y_1),\\
 R_1 &< I(V_1;Y_2,V_2),\\
 R_2 &< I(V_2;Y_2,V_1),\\
 R_1 + R_2 &< I(V_1,V_2;Y_2).
\end{align*}
To achieve arbitrary point here, one can first find a good
point-in-point code for $\Rr_{13}(p)$ and a monotone permutation for the MAC
region $\Rr_{23}(p)$. Then apply method in~\ref{sec:alignment} to align the
incompatible indices in the code for $V_1^N$. This achieves any point in the
rate region $\Rr_{13}(p)\cap \Rr_{23}(p)$. Similarly for each decoding set, one
can design a corresponding 
polar coding scheme based on the method above. Therefore, the proposed polar
coding scheme achieves the optimal rate region given Cover's superposition
encoding. The generalization to $L$-user broadcast channels can be done
similarly.

As a side remark, the independence between $V_1$ and $V_2$ in Cover's
superposition coding is important for transforming the broadcast channel into a
two-user interference channel. For general correlated $(V_1,V_2) \sim
p(v_1,v_2)$ as in Marton coding for broadcast channels, one needs different
techniques. A method for Marton coding as well as an alternative
polar coding scheme for Bergmans's superposition coding can be found
in~\cite{MHSU2014}.

\section{Discussion}
We have shown a polar coding method for the general interference networks that
achieves the optimal rate region when the encoding is restricted to random
coding ensembles with superposition coding and time sharing~\cite{BEK2012}. As
special cases, the method achieves the capacity region of the compound MAC, the
Han--Kobayashi inner bound for two-user
interference channels, and the superposition coding inner bound for broadcast
channels. 

One drawback of the current method is the long blocklength needed for large
scale
networks. When there are $L$ receivers in the networks, one needs to do $L-1$
alignments to resolve the incompatible indices in $L$ permutations, which
makes the blocklength scale with the network size. 

One crucial component in the current method is Ar\i kan's `polar splitting' for
MAC. It would be interesting to compare it to regular rate splitting for MAC as
in~\cite{GRUW2001}. Both schemes achieve optimal
performance in MAC. However, for interference channels, the former, together
with the alignment method,
achieves the best known rate region while the latter is strictly suboptimal
information theoretically~\cite{WSK2014}.

In a parallel study~\cite{WSK2014}, a successive decoding based random coding
scheme is presented, which also achieves Han--Kobayashi inner bound. Some
similarities and connections can be found in the way
the two schemes resolve the incompatibility of the two MACs.

\section*{Acknowledgement}
L.~Wang would like to thank Young-Han Kim for suggesting the problem
and the guidance throughout.  
Her work is supported by the 2013
Qualcomm Innovation Fellowship, the SAMSUNG Global Research Outreach
program, and the National Science Foundation under grant CCF-1116139.
E.~\c Sa\c so\u glu's work is supported by the Swiss National Science
Foundation under grant PBELP2\_137726.

\bibliographystyle{IEEEtran}
\bibliography{polaric}

\end{document}